\documentclass[a4paper, 11pt]{article}

\usepackage{amsmath,amsfonts,amssymb,amsthm,graphicx,graphics,epsf}
\usepackage{graphicx}

\usepackage{algorithmic}
\usepackage{algorithm}
\usepackage{hyperref,color}
\usepackage{xspace}

\usepackage[margin=2cm]{geometry}

\newtheorem{theorem}{Theorem}[section]
\newtheorem{corollary}[theorem]{Corollary}
\newtheorem{lemma}[theorem]{Lemma}

\newtheorem{observation}[theorem]{Observation}

\newcommand{\bd}{\ensuremath{\partial G} }                 

\newcommand{\cvisible}{color-visible\xspace}
\newcommand{\cvisibility}{color-visibility\xspace}

\begin{document}

\title{Bichromatic Compatible Matchings}

\author{Greg Aloupis\thanks{D\'epartement d'Informatique, Universit\'e Libre de Bruxelles, Brussels, Belgium {\tt aloupis.greg@gmail.com,} \tt{\{lbarbafl,slanger\}@ulb.ac.be}}\ \thanks{Charg\'e de recherches du F.R.S.-FNRS.} \and Luis Barba$^*$\thanks{School of Computer Science, Carleton University, Ottawa, Canada}\ \thanks{Boursier FRIA du FNRS} \and Stefan Langerman$^*$\thanks{Directeur de recherches du F.R.S.-FNRS.} \and Diane L. Souvaine\thanks{NSF Grant \#CCF-0830734., Department of Computer Science, Tufts University, Medford, MA.
{\tt dls@cs.tufts.edu}}
}
\date{}

\maketitle
\begin{abstract}
For a set $R$ of $n$ red points and a set $B$ of
$n$ blue points, a \emph{$BR$-matching} is a
non-crossing geometric perfect matching where each segment has one endpoint in $B$ and one in $R$. Two 
$BR$-matchings are compatible if their union is also
non-crossing. We prove that, for any two distinct $BR$-matchings $M$ and $M'$, there exists a sequence
of $BR$-matchings $M = M_1 , \ldots,
M_k = M'$  such that $M_{i-1} $ is compatible with $M_i$.  
This implies the connectivity of the \emph{compatible bichromatic matching graph} containing one node for each  $BR$-matching and an edge joining each pair of compatible $BR$-matchings, thereby
answering the open problem posed by Aichholzer et al. in~\cite{CompatibleMatchingsForPSLG}.
\end{abstract}


\section{Introduction}

A planar straight line graph (PSLG) is a geometric graph in
which the vertices are points embedded in the plane and the edges are non-crossing line segments. 
There are many special types of PSLGs
of which we name a few.  A triangulation is a PSLG to
which no more edges may be added between existing vertices.
A \emph{geometric matching} of a given point set $P$ is a 1-regular PSLG consisting of pairwise
disjoint line segments in the plane joining points of $P$.  
A geometric matching is \emph{perfect} if every point in $P$ belongs to exactly one segment.

Two branches of study on PSLGs include those of geometric
augmentation and geometric reconfiguration.  
A typical augmentation
problem on PSLG $G = (V,E)$ asks
for a set of new edges $E'$ such that the graph $(V,E
\cup E')$ retains or gains some desired properties (see survey by Hurtado and T\'oth~\cite{Hurtado2012}).  

A typical reconfiguration problem on a pair of PSLGs
$G$ and $G'$ sharing some property asks for a sequence of PSLGs
$G = G_0, \ldots, G_k = G'$  where each successive pair of PSLGs  $G_{i-1}$,
$G_i$ jointly satisfy some geometric constraints.
In some situations, a bound on the value of $k$ is desired as well~\cite{Aichholzer200619, Aichholzer20023, Aichholzer2009617, Buchin2009, Hurtado1999, Hurtado1996, Razen2008}.
  
One such solved problem is that of reconfiguring triangulations: given two triangulations $T$ and $T'$, one can compute a sequence of triangulations $T=T_0 ,
\ldots, T_k = T'$ on the same point set such that $T_{i-1}$ can be reconfigured to $T_i$  by flipping one edge.
   Furthermore, bounds on the value of $k$ are known:
$O(n^2)$ edge flips are always sufficient~\cite{Hurtado1996} and
$\Omega (n^2)$ edge flips are sometimes necessary~\cite{Hurtado1999}.


Two PSLGs on the same vertex set are {\it compatible} if their union is planar.
Compatible geometric matchings have been the object of study in both augmentation and reconfiguration problems.  For example, the {\it
Disjoint Compatible Matching Conjecture}~\cite{Aichholzer2009617} was
recently solved in the affirmative~\cite{Ishaque2011}: every perfect planar matching $M$ of $2n$ segments on $4n$
points can be augmented by $2n$ additional segments to form a PLSG
that is the union of simple polygons.


Let $M$ and $M'$ be two perfect planar matchings of a given point set. The reconfiguration problem asks for a \emph{compatible sequence} of matchings $M = M_0 , \ldots, M_k = M'$ such that $M_{i-1}$ is compatible with $M_i$ for all $i \in \{1, \ldots,k\}$.
Aichholzer et al.~\cite{Aichholzer2009617} proved that there is always a compatible sequence of $O(\log n)$ matchings that reconfigures any given matching into a canonical matching. Thus, the
 \emph{compatible matching graph}, that has one node for each perfect planar matching and an edge between any two compatible matchings, is connected with diameter
  $O(\log n)$.
Razen~\cite{Razen2008} proved that the distance between two nodes in this graph is sometimes $\Omega(\log n/ \log \log n)$.


A natural question to extend this research is to ask what happens with bichromatic point sets in which the segments must join points from different colors. 
Let $P= B\cup R$ be a set of points in the plane in general position where $|R|=|B| = n$.  
A straight-line segment with one endpoint in $B$ and one in $R$ is called a \emph{bichromatic segment}.  
A perfect planar matching of $P$ where every segment is bichromatic is called a \emph{$BR$-matching}. 
Sharir and Welzl~\cite{Sharir2006} proved that the number of $BR$-matchings of $P$ is at most $O(7.61^n)$.
Hurtado et al.~\cite{Hurtado200814} showed that any $BR$-matching can be augmented to a crossing-free bichromatic spanning tree in $O(n \log n)$ time.
Aichholzer et al.~\cite{CompatibleMatchingsForPSLG} proved that for any $BR$-matching $M$ of $P$, there are at least $\lceil\frac{n-1}{2}\rceil$ bichromatic segments spanned by $P$ that are compatible with $M$. Furthermore, there are $BR$-matchings with at most $3n/4$ compatible bichromatic segments.

At least one $BR$-matching can always be produced by recursively applying \emph{ham-sandwich cuts}; see Fig.~\ref{fig:CanonicalMatching} for an illustration. A $BR$-matching produced in this way is called a \emph{ham-sandwich matching}.
Notice that the general position assumption is sometimes necessary to guarantee the existence of a $BR$-matching.
However, not all $BR$-matchings can be produced using ham-sandwich cuts. 
Furthermore, some point sets admit only one $BR$-matching, which must be a ham-sandwich matching.

Two $BR$-matchings $M$ and $M'$ are \emph{connected} if there is a sequence of $BR$-matchings $M = M_0, \ldots, M_k = M'$, such that $M_{i-1}$ is compatible with $M_{i}$, for $1\leq i\leq k$.
An open problem posed by Aichholzer et al.~\cite{CompatibleMatchingsForPSLG} was to prove that all $BR$-matchings of a given point set are connected\footnote{This problem was also posed during the EuroGIGA meeting that took place after EuroCG 2012.}.  We answer this in the affirmative by using a ham-sandwich matching $H$ as a canonical form. 
Consider the first ham-sandwich cut line $\ell$ used to construct $H$. We show how to reconfigure any given $BR$-matching via a compatible sequence, so that the last matching in the sequence contains no segment crossing $\ell$.
We use this result recursively, on every ham-sandwich cut used to generate $H$, to show that any given $BR$-matching is connected with $H$.

\section{Ham-sandwich matchings}\label{Section:Ham-Sandwich}
In this paper, a \emph{ham-sandwich cut} of $P$ is a line passing through no point of $P$ and containing exactly $\lfloor \frac{n}{2}\rfloor$ blue and $\lfloor \frac{n}{2}\rfloor$ red points to one side.
Notice that if $n$ is even, then this matches the \emph{classical} definition of ham-sandwich cuts (see Chapter 3 of ~\cite{MatousekBorsukUlam}). However, when $n$ is odd, a ham-sandwich cut $\ell$ according the classical definition will go through a red and a blue point of $P$. In this case, we obtain a ham-sandwich cut, according to our definition, by slightly moving $\ell$ away from these two points without changing its slope and without reaching another point of $P$. By the general position assumption this is always possible.

Since every bichromatic point set admits a ham-sandwich cut, $P$ admits at least one $BR$-matching resulting from recursively applying ham-sandwich cuts. We call this a \emph{ham-sandwich matching}; see Fig.~\ref{fig:CanonicalMatching}. Notice that $P$ may admit several ham-sandwich matchings.

\begin{figure}[tb]
\centering
\includegraphics{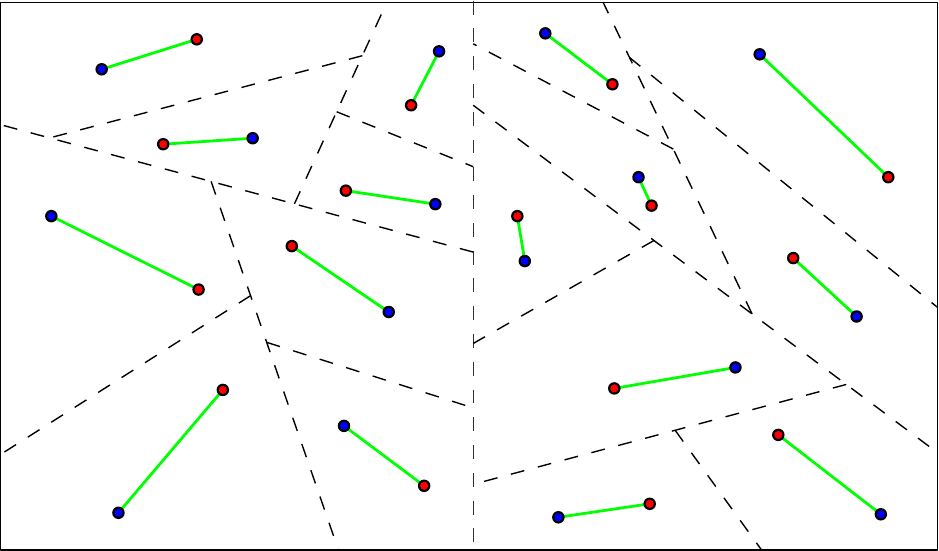}
\caption{A ham-sandwich matching obtained by recursively applying ham-sandwich cuts.}
\label{fig:CanonicalMatching}
\end{figure}

Let $M$ be a $BR$-matching of $P$.
In this section we prove that $M$ is connected with a ham-sandwich matching $H$ of $P$.
Consider a ham-sandwich cut $\ell$ used to construct $H$. The idea of the proof is to show the existence of a $BR$-matching $M'$, compatible with $M$, such that $M'$ has ``fewer'' crossings with $\ell$ according to some measure defined on $BR$-matchings. 
By repeatedly applying this result, we end up with a matching $M^\ell$, connected with $M$, such that no segment of $M^\ell$ crosses $\ell$.
Once we know how to avoid a ham-sandwich cut, we can apply the same result recursively on every ham-sandwich cut used to generate $H$. In this way, we obtain a sequence of compatible $BR$-matchings that connect $M$ with $H$.

The main ingredient to obtain these results is Lemma~\ref{lemma:CompatibleBRMatching}.
Before stating this lemma, we need a few more definitions.

Given a line $\ell$ that contains no point of $P$, let $S_{M,\ell}$ be the set of segments of $M$ that cross $\ell$.
We say that $\ell$ is a  \emph{chromatic cut} of $M$ if $|S_{M,\ell}| \geq 2$ and not all endpoints of $S_{M,\ell}$ on one side of $\ell$ have the same color. Without loss of generality, we can assume that if a chromatic cut $\ell$ exists, then it is vertical and no segment of $M$ is parallel to $\ell$. 
The following observation shows the relation existing between chromatic and ham-sandwich cuts.

\begin{lemma}\label{lemma:HamCutsAreChromaticCuts}
Given any  $BR$-matching $M$ of $P$, every ham-sandwich cut of $P$ either crosses no segment of $M$, or is a chromatic cut of $M$.
\end{lemma}
\begin{proof}
Let $\ell$ be a ham-sandwich cut of $M$. Recall that the number of blue and red points to the right of $\ell$ must be the same. Moreover, every segment that is not crossed by $\ell$ has both of its endpoints on one of its sides.
Therefore, if $\ell$ crosses a segment of $M$ having a red endpoint to the right of $\ell$, then, to maintain the balance of red and blue points, $\ell$ must cross another segment having a blue endpoint to the right of $\ell$. That~is, $\ell$ is a chromatic cut of~$M$.
\end{proof}

We say that a point lies above (\emph{resp.} below) a segment if it lies above (\emph{resp.} below) the line extending that segment. We proceed to state Lemma~\ref{lemma:CompatibleBRMatching}. However, its proof is deferred to Section~\ref{Section:Avoiding Chromatic Cuts} for ease of readability.

\begin{lemma}\label{lemma:CompatibleBRMatching}
Let $M$ be a $BR$-matching of $P$ and let $\ell$ be a chromatic cut of $M$.
There exists a $BR$-matching $M'$ of $P$, compatible with $M$, with the following properties. 
There is a segment $s$ in $M\setminus M'$ that crosses $\ell$ such that all segments of $M$ that cross $\ell$ below $s$ belong also to $M'$. Moreover, these are the only segments of $M'$ crossing $\ell$ below $s$.
\end{lemma}

In other words, Lemma~\ref{lemma:CompatibleBRMatching} states that we can find a $BR$-matching $M'$, compatible with $M$, such that a segment $s$ from $M$ that crosses $\ell$ does not appear in $M'$. Moreover, every segment of $M$ that crosses $\ell$ below $s$ is preserved in $M'$, and no new segment crossing $\ell$ below $s$ is introduced. However, we have no control over what happens above $s$.

\begin{lemma}\label{lemma:CompatibleSequence}
Given a $BR$-matching $M$ of $P$ and a ham-sandwich cut $\ell$, 
there is a $BR$-matching $M^\ell$ connected with $M$ such that no segment of $M^\ell$ crosses $\ell$.
\end{lemma}
\begin{proof}
Assume that $\ell$ is a chromatic cut of $M$, i.e., that it crosses at least one segment of $M$. Otherwise, the result follows trivially.
Given a $BR$-matching $W$ of $P$ such that $\ell$ is a chromatic cut of $W$, let $\textsc{Next}(W)$ be the matching, compatible with $W$, that exists as a consequence of Lemma~\ref{lemma:CompatibleBRMatching} when applied on $W$.

Let $M_0 = M$. If $S_{M_i, \ell} \neq \emptyset$, i.e., there are segments of $M_i$ that crosses~$\ell$,
then let $M_{i+1} = \textsc{Next}(M_i)$. We claim that the sequence $\varphi = (M_0, M_1,\ldots, M_h)$ is finite and hence, that $M^\ell := M_h$ has no segment that crosses $\ell$.
Note that the sequence $\varphi$ is well defined by Lemma~\ref{lemma:HamCutsAreChromaticCuts}.

Assume without loss of generality that $\ell$ is a vertical line.
Let $\mathcal C_P = \{z_1, z_2, \ldots, z_m\}$ be the set of all possible $O(n^2)$ bichromatic segments that cross $\ell$. Assume that the segments of $\mathcal C_P$ are sorted, from bottom to top, according to their intersection with $\ell$.
Given a $BR$-matching $W$ of $P$, let $\chi_{_W} = b_1b_2\ldots b_m$ be a binary number where $b_i$ is defined as follows:
$$b_i = \left\{
\begin{array}{lll}
1&&\text{If $z_i$ belongs to $W$}\\
0&&\text{Otherwise}
\end{array}\right.$$

Let $M_i$ and $M_{i+1}$ be two consecutive matchings in $\varphi$.
By Lemma~\ref{lemma:CompatibleBRMatching}, there is a segment $s$, corresponding to some segment $z_k$ in $\mathcal C_P$, such that $s= z_k$ belongs to $M_i$ but not to $M_{i+1}$. 
Moreover, if $z_j$ is a segment that crosses $\ell$ below $z_k$, then $z_j$ belongs to $M_i$ if and only if $z_j$ belongs to $M_{i+1}$.
Therefore, the $k$-th digit of $\chi_{_{M_i}}$ is 1 while the $k$-th digit of $\chi_{_{M_{i+1}}}$ is 0. 
Moreover, for every $j< k$, the $j$-th digit of $\chi_{_{M_i}}$ is identical to the $j$-th digit of $\chi_{_{M_{i+1}}}$.
This implies that $\chi_{_{M_i}} > \chi_{_{M_{i+1}}}$. Therefore, $\Phi =\chi_{_{M_0}}, \chi_{_{M_1}}, \ldots, \chi_{_{M_h}}$ is a strictly decreasing sequence.
This means that no $BR$-matching is repeated and that $\Phi$ converges to zero yielding our claim, i.e., we reach a $BR$-matching containing no segment that crosses $\ell$.
\end{proof}

\begin{theorem}\label{theorem:The graph is connected}
Let $P= B\cup R$ be a set of points in the plane in general position where $|R|=|B|$. 
If $M$ is a $BR$-matching of $P$ and $H$ is a ham-sandwich matching of $P$, then $M$ and $H$ are connected, i.e., there is a sequence of $BR$-matchings $M = M_0, \ldots, M_r = H$ such that $M_{i-1}$ is compatible with $M_i$ for $1\leq i\leq r$.
\end{theorem}

\begin{proof}
The proof goes by induction on the size of $P$. 
Notice that the result follows trivially if $|P| = 2$ as it contains a unique $BR$-matching.

Assume that the result holds for any bichromatic point set with fewer than $n$ points. 
Let $\ell$ be the first ham-sandwich cut line used to construct $H$.
By Lemma~\ref{lemma:CompatibleSequence}, there is a matching $M^\ell$ such that $M$ and~$M^\ell$ are connected, and no segment of $M^\ell$ crosses $\ell$.
Let $\Pi_1$ and $\Pi_2$ be the two halfplanes supported by $\ell$. 
For~$i\in \{1,2\}$, let $P_i$ be the set of points of $P$ that lie in $\Pi_i$ and let $M_i$ and $H_i$ be, respectively, the set of segments of $M^\ell$ and $H$ that are contained in $\Pi_i$.
Because $|P_i| < |P|$,  $M_i$ and $H_i$ are connected by the induction hypothesis.

Since every $BR$-matching of $P_1$ is compatible with every $BR$-matching of $P_2$, we can merge the two compatible sequences obtained by the 
recursive construction that certify that $M_i$ and $H_i$ are connected.
Thus, $M^\ell$ is connected with $H$ and because $M$ is connected to $M^\ell$, $M$ and $H$ are also connected. 
\end{proof}

Let $V$ be the set of all $BR$-matchings of $P$ and let $G_P$ be the \emph{compatible bichromatic matching graph} of~$P$ with vertex set $V$, where there is an edge between two vertices if their corresponding $BR$-matchings are compatible. 

\begin{corollary}\label{corollary:G_P is connected}
Given a set of points $P= B\cup R$ in general position such that $|R|=|B| = n$,
the graph $G_P$ is connected.
\end{corollary}

While Theorem~\ref{theorem:The graph is connected} implies the connectedness of $G_P$, proving non-trivial upper bounds on its diameter remains elusive. Lemma~\ref{lemma:Lower Bound}, depicted in Fig.~\ref{fig:LowerBound}, provides a linear lower bound on the diameter of $G_P$. Since Sharir and Welzl~\cite{Sharir2006} proved that the number of vertices in $G_P$ is at most $O(7.61^n)$, we obtain a trivial exponential bound on its diameter. 

\begin{figure}[htb]
\centering
\includegraphics{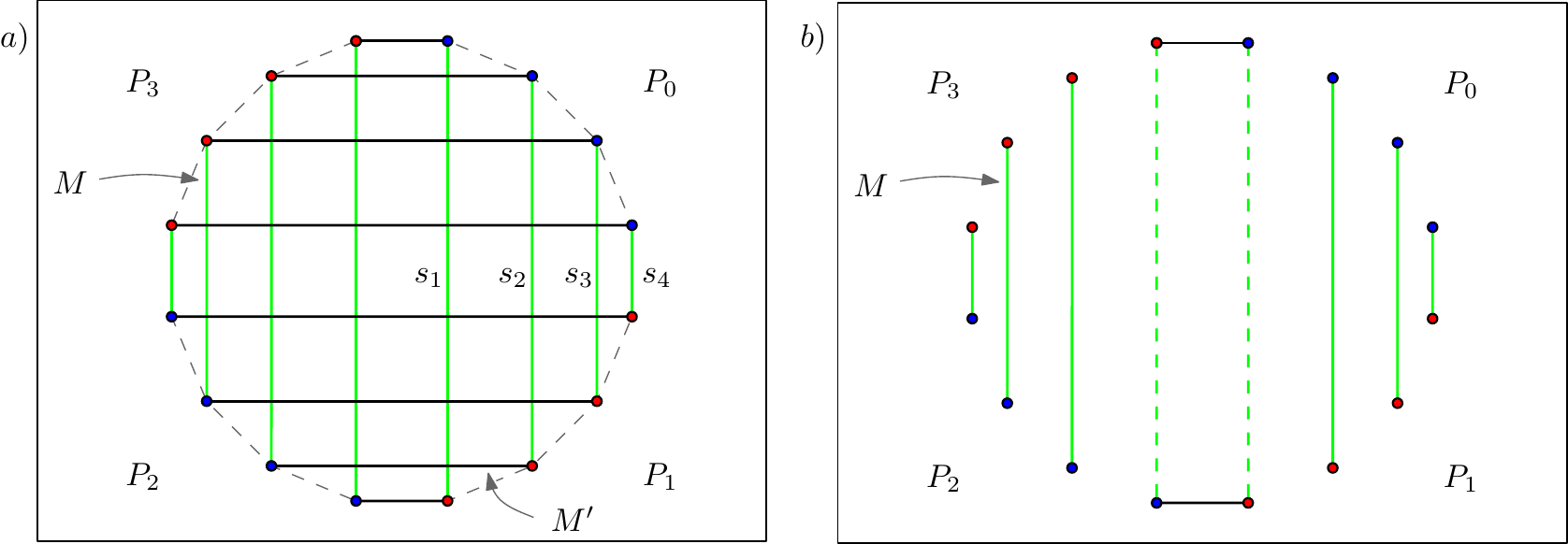}
\caption{$a)$ Two $BR$-matchings $M$ and $M'$ at distance $\Omega(n)$ in the graph $G_P$ that are both ham-sandwich matchings. $b)$ The only two segments compatible with $M$ are the topmost and bottommost segments of $M'$.}
\label{fig:LowerBound}
\end{figure}

\begin{lemma}\label{lemma:Lower Bound}
There exists a bichromatic set $P = B\cup R$ of $4n$ points that
admits two $BR$-matchings at distance $\Omega(n)$ in  $G_P$.
\end{lemma}
\begin{proof}
Let $P$ be the set of vertices of a regular $4n$-gon $Q$. Partition $P$ into four disjoint sets $P_0, \ldots, P_3$, each of $n$ consecutive points along the boundary of $Q$. Assume without loss of generality that the unique edge  on the boundary of $Q$ joining a point from $P_3$ with a point in $P_0$ is parallel to the $x$-axis. Moreover, assume that this edge is the topmost edge of $Q$.
Let $B = P_0\cup P_2$ and let $R = P_1\cup P_3$.
Note that for any $0\leq i\leq 3$, both bichromatic point sets $P_i\cup P_{i+1}$ and $P_i\cup P_{i-1}$ have unique $BR$-matchings   where sum is taken modulo~$3$; see Fig.~\ref{fig:LowerBound}$(a)$.

Let $M$ be the union of the unique $BR$-matchings of $P_0\cup P_1$ and $P_2\cup P_3$. Analogously, let $M'$ be the union of the $BR$-matchings of $P_1\cup P_2$ and $P_0\cup P_3$; see Fig.~\ref{fig:LowerBound}$(a)$ for an illustration.
Let $p$ be a point of $P$ and let $M(p)$ and $M'(p)$ be the points matched with $p$ in $M$ and $M'$, respectively.
We claim that in any $BR$-matching $W$ of $P$, the point $p$ is matched with either $M(p)$ or $M'(p)$. 
If this is not the case, then the segment $s$ joining $p$ with its neighbor in $W$ separates $P$ into two sets, one of them having an unbalanced number of blue and red points---a contradiction as the points on each side of $s$ support a $BR$-matching that does not cross $s$.

Let $s_1, \ldots s_{n}$ be the last $n$ segments of $M$ when sorted from left to right. 
We claim that to connect~$M$ with a $BR$-matching that doesn't contain $s_i$, we need a compatible sequence of $BR$-matchings of length at least $i$. The proof goes by induction on $i$. 

For the base case, note that only the topmost and bottommost segments of $M'$ are compatible with $M$, any other bichromatic segment either crosses an edge of $M$ or splits $P$ into two unbalanced sets. By adding these two edges and removing the segments of $M$ that form a cycle with them, we obtain the unique $BR$-matching compatible with $M$ where $s_1$ is not present; see Fig.~\ref{fig:LowerBound}$(b)$

Assume that the result holds for any $j<i$. Let $p$ be an endpoint of $s_i$ and note that the segment~$uM'(p)$ crosses $s_{i-1}$. 
That is, $p$ is matched through $s_i$ in any $BR$-matching that contains~$s_{i-1}$.
Consequently, to reach a $BR$-matching where $s_i$ is not present, we need to first remove $s_{i-1}$, which requires at least $i-1$ steps by the induction hypothesis, and then at least one step more to remove $s_i$. Therefore, a sequence of at least~$i$ $BR$-matchings is needed to connect $M$ with a $BR$-matching that doesn't contain $s_i$.

Thus, to connect $M$ with $M'$ where $s_n$ is not present, we need a compatible sequence of $\Omega(n)$ $BR$-matchings proving our result.
\end{proof}

\section{Well-colored graphs and basic tools}\label{Section:Tools}
In this section, we introduce some tools that will help us prove Lemma~\ref{lemma:CompatibleBRMatching} in Section~\ref{Section:Avoiding Chromatic Cuts}.

Given a face $F$ of a PSLG,
we denote its interior by $int(F)$ and its boundary by $\partial F$. 
In the remainder, we will only consider bounded faces when we refer to a face of a PSLG.
A vertex $v$ is \emph{reflex in $F$} if there is a non-convex connected component in the intersection of $int(F)$ with any disk centered at $v$.
Notice that a vertex can be reflex in at most one face of a PSLG. 
A vertex of a PSLG is \emph{reflex} if it is reflex in one of its bounded faces.

Let $F$ be a face of a given PSLG whose reflex vertices are colored either blue or red. 
We say that $F$ is \emph{well-colored} if the sequence of reflex vertices along its boundary alternates in color. In the same way, a PSLG is well-colored if all its faces are well-colored. Since a vertex is reflex in at most one face, a well-colored PSLG has an even number of reflex vertices.

Let $G$ be a well-colored PSLG. The \emph{boundary of $G$}, denoted by $\bd$, is the union of all the edges in $G$. The interior of $G$ is the union of the interior of its faces.
Let $M$ be a $BR$-matching such that all segments of $M$ are contained in the interior of $G$, i.e., no segment of $M$ crosses and edge of $G$. 
We show how to \emph{glue} the segments of $M$ to $G$ in such a way that every endpoint of $M$ becomes a reflex vertex. We then provide a technique to construct a new $BR$-matching, compatible with $M$, by matching the reflex vertices of $G$ after the gluing.

\subsection{Coloring a PSLG}

We say that two points $x$ and $y$ are \emph{visible} if the open segment $(x,y)$ is contained in the interior of $G$ and crosses no segment of $M$.
Throughout this paper, the color of a point not in $P$, either on a bichromatic segment or on $\bd$, depends on the position from which it is viewed; see Fig.~\ref{fig:ColoringTheBoundary} for an example.

Assume that $F$ is a well-colored face of $G$ and let $x$ be a point on $\partial F$. 
Let $y$ be a point visible from $x$.
Walk in a straight line from $y$ towards $x$ and make a left turn when reaching $x$, following the boundary of $F$ counterclockwise until reaching a reflex vertex $r$ (if $x$ is reflex, then $r=x$). We say that $x$ is blue (\emph{resp.} red) when viewed from $y$ if $r$ is blue (\emph{resp.} red). If $F$ contains no reflex vertex, then the color of $x$ when viewed from $y$ can be arbitrarily chosen to be blue or red.

This coloring scheme can be used for segments as well. 
For $r\in R$ and $b\in B$, let $x$ be a point in the interior of the bichromatic segment $s = [r,b]$. Let $y$ be a point in the plane, not collinear with $r$ and $b$.
We say that $x$ is \emph{blue} when viewed from $y$ if the triple $y,x,b$ makes a left turn. Otherwise, we say that $x$ is \emph{red} when viewed from $y$. 

Let $z$ and $z'$ be two points such that each one lies either on $\bd$ or on a segment of $M$.
We say that $z$ and $z'$ are \emph{\cvisible} if they are visible and the color of $z$ when viewed from $z'$ is equal to the color of $z'$ when viewed from $z$.

\begin{figure}[htb]
\centering
\includegraphics{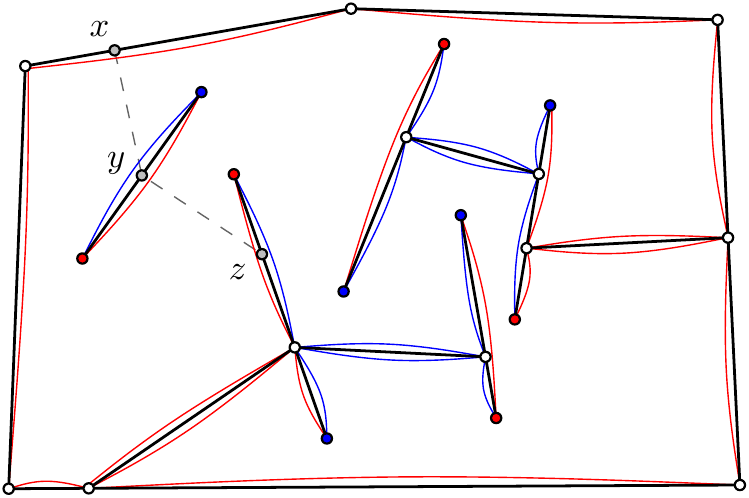}
\caption{The coloring of the boundary points of a PSLG, as well as of a bichromatic segment. The point $y$ is blue when viewed from $x$ but red when viewed from $z$. Moreover, $y$ and $z$ are \cvisible.}
\label{fig:ColoringTheBoundary}
\end{figure}

\subsection{Basic operators for well-colored PSLGs}

Let $z$ be a point in the interior of a segment $s = [a,b]$ of $M$ and let $z'$ be a non-reflex point on $\bd$ such that $z$ and $z'$ are \cvisible.
Operator $\textsc{Glue}$ produces a new PSLG by attaching $s$ to $\partial G$ using $z$ and $z'$ as points of attachment.
Formally, if $z'$ is not a vertex of $G$, then insert it as a vertex by splitting the edge of $G$ that contains $z'$.
Add the vertices $z, a $ and $b$ and the edges $[z,z']$, $[z,a]$ and $[z,b]$ to $G$. 
In the resulting PSLG, denoted by $\textsc{Glue}(G,z, z')$, $a$ and $b$ are both reflex vertices of degree one; see Fig.~\ref{fig:GlueCutOperator}.

Let $y$ and $y'$ be two \cvisible points on $\bd$ such that neither $y$ nor $y'$ are reflex vertices.
Operator $\textsc{Cut}$ joins $y$ with $y'$ in the following way.
Let $F$ be the face of $G$ that contains the segment $[y,y']$. 
If either $y$ or $y'$ is not a vertex of $G$, insert it by splitting the edge where it lies on.
Thus, $[y,y']$ is a chord of $F$, so adding the edge $[y,y']$ to $G$ forms two cycles and splits $F$ into two new faces.
In this way, we obtain a new PSLG $\textsc{Cut}(G, y, y')$ with one face more than $G$; see Fig.~\ref{fig:GlueCutOperator} for an illustration of this operation.

Since both operators join two points by adding the edge between them, we can define an operator $\textsc{GlueCut}$ on $G,z$ and $z'$, that behaves like $\textsc{Glue}$ when $z$ belongs to a segment in $M$, or behaves like $\textsc{Cut}$ if both $z$ and $z'$ belong to $\bd$. The PSLG output by this operator is denoted by $\textsc{GlueCut}(G, z,z')$.

A \emph{Glue-Cut Graph (GCG)} is a well-colored PSLG where every reflex vertex has degree one. Although this definition is more general, we can think of a GCG as a PSLG obtained by repeatedly applying $\textsc{GlueCut}$ operations between a convex polygon and the segments of a $BR$-matching contained in its interior; see Fig.~\ref{fig:GlueCutOperator} for an example.

\begin{figure}[htb]
\centering
\includegraphics{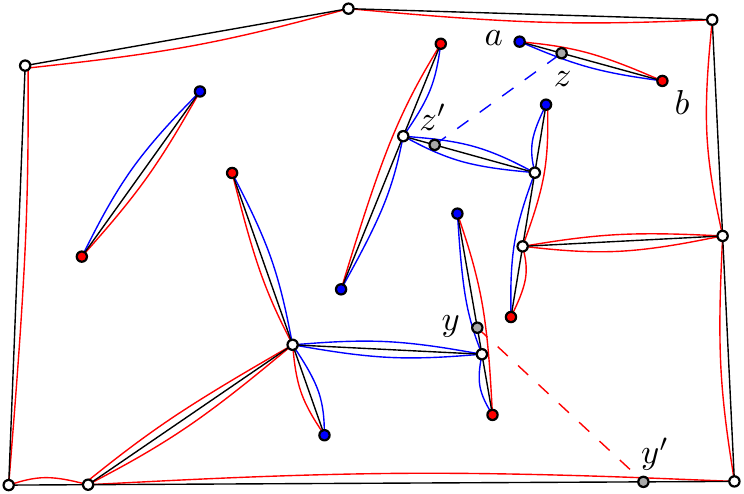}
\caption{Two pairs of \cvisible points $z,z'$ and $y,y'$, where $z$ and $z'$ can be joined by the $\textsc{Glue}$ operator and $y$ and $y'$ by the $\textsc{Cut}$ operator.}
\label{fig:GlueCutOperator}
\end{figure}

\begin{lemma}\label{lemma:Closed Under GlueCut}
The family of Glue-Cut Graphs is closed under the $\textsc{GlueCut}$ operator. 
\end{lemma}
\begin{proof}
Let $G$ be a GCG and $z$ be a point in a bichromatic segment $s$ contained in the interior of $G$. Let $z', y$ and $y'$ be points on $\partial G$ such that $z$ and $z'$ (\emph{resp.} $y$ and $y'$) are \cvisible.
When constructing $\textsc{Glue}(G,z,z')$, the endpoints of $s$ become reflex vertices of degree one. 
That is, we add one red and one blue reflex vertex to $G$. Therefore, to prove that $\textsc{Glue}(G,z,z')$ is well-colored, it suffices to show that the points are added in the correct order which is guaranteed by the \cvisibility of $z$ and $z'$; see Fig.~\ref{fig:GlueCutOperator}.

On the other hand, $\textsc{Cut}(G, y, y')$ neither adds nor removes reflex vertices of $G$. 
This operation divides a well-colored face of $G$ into two, by inserting a new edge.
Consider either of the           new faces.  Let $a,b$ be the first reflex vertices found when following the boundary from this edge on each side.   Since $y$ and $y'$ are \cvisible when 
$\textsc{Cut}$ is invoked, we know that $a$ and $b$ are of different color.  Thus each new face, and therefore $\textsc{Cut}(G, y, y')$,
is well-colored; see Fig.~\ref{fig:GlueCutOperator}.
\end{proof}

\subsection{Simplification of a GCG}

Let $F = (v_1,v_2,\ldots,v_k,v_1)$ be a face of a GCG given as a sequence of its vertices in clockwise order along its boundary.
For each vertex $v_i$, if the triple $v_{i-1},v_i,v_{i+1}$ makes a right turn, let $x_i$ be a point at distance $\varepsilon>0$ from $v_i$, lying on the bisector of the convex angle formed by $[v_{i-1},v_i]$ and $[v_i,v_{i+1}]$. If $v_i$ is reflex in $F$, let $x_i = v_i$.
Otherwise, if $v_{i-1},v_i,v_{i+1}$ are collinear, do nothing.
Let $\mathcal P_F=(x_1, \ldots, x_k, x_1)$ (consider only the indices where $x_i$ is defined).
By choosing $\varepsilon$ sufficiently small, $\mathcal P_F$ is a simple polygon contained in $F$ such that every reflex vertex $v_j$ in $F$ remains reflex in $\mathcal P_F$ and no reflex vertex is created; see Fig.~\ref{fig:SimpleExtension}.
We call $\mathcal P_F$ a \emph{simplification} of $F$. Though the simplification of a face $F$ is not unique as it depends on the choice of $\varepsilon$, the results presented in this paper hold for any simplification. Therefore, when alluding to $\mathcal P_F$, we refer to any simplification of $F$.

\begin{observation}\label{obs:SimpleExtension}
For every bounded face $F$ of a GCG,  $\mathcal P_F$ is a simple polygon contained in $F$ 
such that $F$ and $\mathcal P_F$ share the same set of reflex vertices.
\end{observation}

\begin{figure}[htb]
\centering
\includegraphics{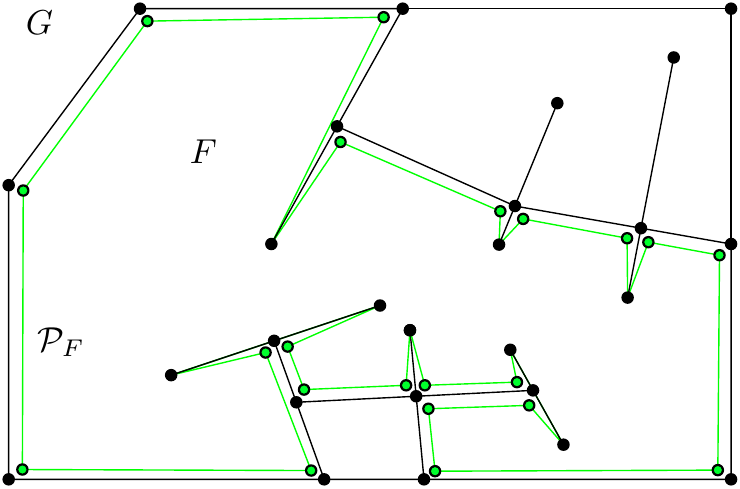}
\caption{A face $F$ of a GCG $G$ and its simplification $\mathcal P_F$, contained in $F$, with the same set of reflex vertices.}
\label{fig:SimpleExtension}
\end{figure}

Let $F_1, \ldots, F_k$ be the bounded faces of a GCG $G$.
We call $\mathcal P_G = \bigcup \mathcal P_{F_i}$ the \emph{simplification of $G$}. 
Note that $\mathcal P_G$ is the union of a set of disjoint simple polygons.

\subsection{Merging a matching with a GCG}\label{section:GluingTheSegments}

\begin{lemma}\label{lemma:AbellanasAlgorithm}
(Rephrasing of Lemma 5 of~\cite{Abellanas2008220}) Let $\mathcal P$ be a simple polygon with an even number of reflex vertices.
There exists a perfect planar matching $M$ of the reflex vertices of $\mathcal P$, such that each segment of $M$ is contained in $\mathcal P$ (or on its boundary). 
\end{lemma}

Let $C = \{r_0, \ldots, r_k\}$  be the set of reflex vertices of a simple polygon $\mathcal P$ sorted along its boundary. 
Let $M$ be a perfect planar matching of $C$ that exists by Lemma~\ref{lemma:AbellanasAlgorithm}.
Let $[r_i, r_j]$ be a segment of $M$, and note that this segment splits $\mathcal P$ into two sub-polygons. Notice that if $[r_i,r_j]$ is contained in the boundary of $\mathcal P$, then one sub-polygon is a segment and the other one is $\mathcal P$ itself. 
In order for $M$ to be perfect and planar, 
each sub-polygon must contain an even number of reflex vertices. Therefore, if a segment $[r_i, r_j]$ belongs to $M$, then $i \bmod 2 \neq j \bmod 2$. This implies that if $\mathcal P$ is well-colored, then $M$ is a $BR$-matching.

The main tool to construct $BR$-matchings of the reflex vertices of a GCG comes from the following lemma; see Fig.~\ref{fig:ReflexBRMatching} for an illustration.

\begin{figure}[htb]
\centering
\includegraphics{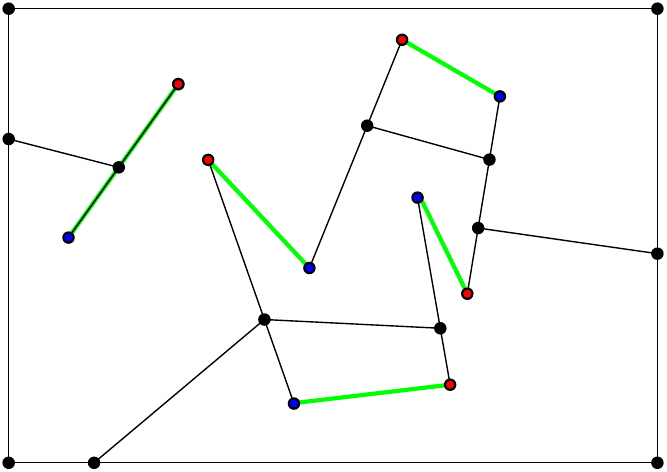}
\caption{A GCG and a $BR$-matching of its reflex vertices.}
\label{fig:ReflexBRMatching}
\end{figure}

\begin{lemma}\label{lemma:ReflexMatching}
If $G$ is a GCG, then there is a $BR$-matching $M$ of the reflex vertices of $G$, such that each segment of $M$ is contained in $\mathcal P_G$ (or on its boundary).
\end{lemma}
\begin{proof}
Let $F_1, \ldots, F_k$ be the well-colored faces of $G$. 
By Observation~\ref{obs:SimpleExtension}, each $F_i$ and its simplification $\mathcal P_{F_i}$ share the same set of reflex vertices.
By Lemma~\ref{lemma:AbellanasAlgorithm}, there is a matching $M_i$ of the reflex vertices of $\mathcal P_{F_i}$, such that each segment lies either in the interior or on the boundary of $\mathcal P_{F_i}$. 
Since $F_i$ is well-colored, $M_i$ is a $BR$-matching.
Note that a vertex can be reflex in at most one face of $G$.
Therefore, $M = \bigcup M_i$ is a $BR$-matching of the reflex vertices of $G$ and each segment of $M$ lies either in the interior or on the boundary of $\mathcal P_G$.
\end{proof}

\subsection{Gluing BR-matchings}

Let $X$ be a GCG and let $M$ be a $BR$-matching contained in the interior of $X$.
In this section, we show how to glue the segments of $M$ to the boundary of $X$. In this way, we obtain a GCG $G$ such that the endpoints of the segments of $M$ are all reflex vertices of $G$. 
Thus, by Lemma~\ref{lemma:ReflexMatching}, we can obtain a $BR$-matching $M'$ of the reflex vertices of $G$ where every segment is contained in $\mathcal P_G$, i.e. we can obtain a $BR$-matching $M'$ whose union with $M$ contains no crossings.

Assume that the vertices of $X$ and the endpoints of $M$ are in general position and that no two points have the same $x$-coordinate.

Let $s$ be the segment with the rightmost endpoint among all segments of $M$.
We may assume that the left (\emph{resp.} right) endpoint of $s$ is blue (\emph{resp.} red) and hence, that $s$ is blue (\emph{resp.} red) when viewed from below (\emph{resp.} above).

Extend $s$ to the right until it reaches the interior of a segment $s'$ on $\partial X$ at a point $y$.
Depending on the color of $s'$ when viewed from $s$, choose a point $y'$ in the interior of $s'$ above (\emph{resp.} below) $y$ if $s'$ is red (\emph{resp.} blue). 
Choose $y'$ sufficiently close to $y$ so that the whole segment $s$ is visible from $y'$. 
This is always possible because $y$ is visible from the right endpoint of $s$.
Let $m$ be the midpoint of $s$ and note that $m$ and $y'$ are \cvisible by construction. 
Let $X' = \textsc{Glue}(X, m, y')$ and note that by Lemma~\ref{lemma:Closed Under GlueCut}, $X'$ is a GCG. Moreover, the endpoints of $s$ become reflex vertices of $X'$; see Fig.~\ref{fig:GluingTheSegments}. 
Remove $s$ from $M$, let $X = X'$ and repeat this construction recursively until $M$ is empty. We obtain the following result.

\begin{figure}[htb]
\centering
\includegraphics{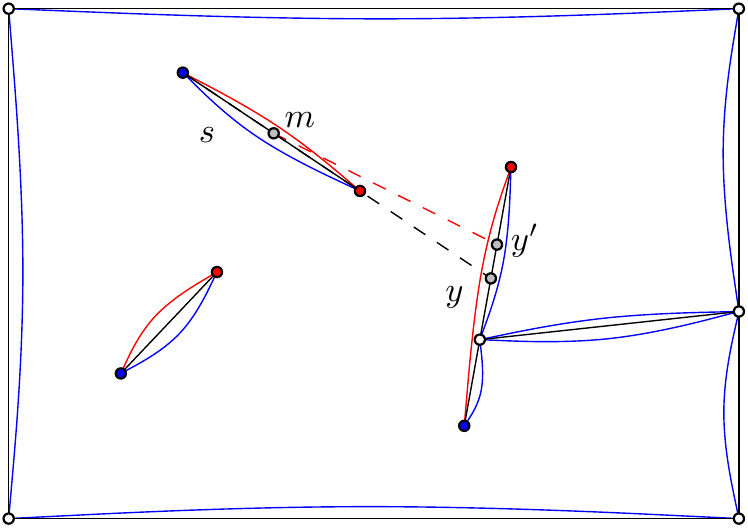}
\caption{Gluing a bichromatic segment with the boundary of a GCG.}
\label{fig:GluingTheSegments}
\end{figure}

\begin{lemma}\label{lemma:Gluing the segments}
Let $X$ be a GCG and $M$ be a $BR$-matching contained in the interior of $X$. 
There is a GCG $G$ augmenting $X$ such that all reflex vertices of $X$ and all endpoints in $M$ are reflex vertices in $G$.
Moreover, every segment of $M$ is contained in $\partial G$.
\end{lemma}

\section{Augmented matchings}\label{Section:Avoiding Chromatic Cuts}
In this section, we provide the proof of Lemma~\ref{lemma:CompatibleBRMatching} presented in Section~\ref{Section:Ham-Sandwich}.

Let $M$ be a $BR$-matching of $P$ and let $\ell$ be a chromatic cut of $M$. Recall that $S_{M,\ell}$ denotes the set of segments of $M$ that cross $\ell$.
We show that it is possible to obtain a new $BR$-matching $M'$ with at least one segment $s$ of $S_{M,\ell}$ absent.
Furthermore, when examining segments of $M$ that cross $\ell$ below $s$, all segments of $S_{M,\ell}$ are preserved in $M'$ and no new segments are introduced.

A vertex $v$ of a GCG $X$ is \emph{isolated} if no line through $v$, intersecting the interior of $X$, supports a closed halfplane containing all the neighbors of $v$. The following observation is depicted in Fig.~\ref{fig:IsolatedVertex}.

\begin{observation}\label{obs:ConvexVertex}
If $v$ is an isolated vertex of a GCG $X$, then $v$ lies outside of $\mathcal P_X$. 
Moreover, if $v'$ is a vertex of $X$ lying outside $\mathcal P_X$, then the open segments joining $v'$ with its neighbors also lie outside of $\mathcal P_X$.

\end{observation}

\begin{figure}[htb]
\centering
\includegraphics{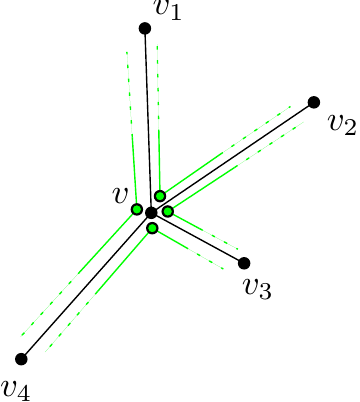}
\caption{An isolated vertex $v$ lying outside of the simplifications of each of its adjacent faces.}
\label{fig:IsolatedVertex}
\end{figure}

Let $\ell$ be a chromatic cut of $M$ and assume that  $S_{M,\ell}= \{s_1, \ldots, s_k\}$ is sorted from bottom to top according to the intersection, $x_i$, of $s_i$ with $\ell$.

The idea of the proof is to construct a GCG $X$ augmenting $M$, using the \textsc{GlueCut} operation and then Lemma~\ref{lemma:Gluing the segments}, in such a way that $x_1, \ldots, x_j$ become isolated vertices of $X$ for some $1\leq j\leq k$. Furthermore, we require $X$ to contain the edge between $x_i$ and $x_{i+1}$ for every $1\leq i < j$.
By Observation~\ref{obs:ConvexVertex}, these edges will lie outside of $\mathcal P_X$ and so will the portion of $\ell$ lying below $s_j$. Thus, this portion of $\ell$ will not be crossed by a $BR$-matching, compatible with $M$, obtained from Lemma~\ref{lemma:ReflexMatching} when applied on $X$. 

Notice that we can only glue $x_i$ with $x_{i+1}$ if they are \cvisible. Although the following lemma shows that there is at least one pair of consecutive \cvisible points among $x_1, \ldots, x_k$, we may not be able to glue all of them. Thus, we will resort to a different strategy that allows us to alter the color of a segment.

\begin{lemma}\label{lemma:ChangeInColor}
There exist two consecutive segments $s_i$ and  $s_{i+1}$ in $S_{M,\ell}$ such that $x_i$ and $x_{i+1}$ are \cvisible.
\end{lemma}
\begin{proof}
Because $\ell$ is a vertical chromatic cut, there exist two segments $s_j$ and $s_h$ in $S_{M,\ell}$ such that the left endpoint of $s_j$ is of different color than the left endpoint of $s_h$.
Therefore, two consecutive segments $s_i$ and $s_{i+1}$ must exist having left endpoints of different color.
This implies that the color of $s_i$ when viewed from above is the same as the color of $s_{i+1}$ when viewed from below. Finally, since $s_i$ and $s_{i+1}$ are consecutive segments in $S_{M,\ell}$, $x_i$ and $x_{i+1}$ are visible.
\end{proof}


\begin{figure}[h!tb]
\centering
\includegraphics{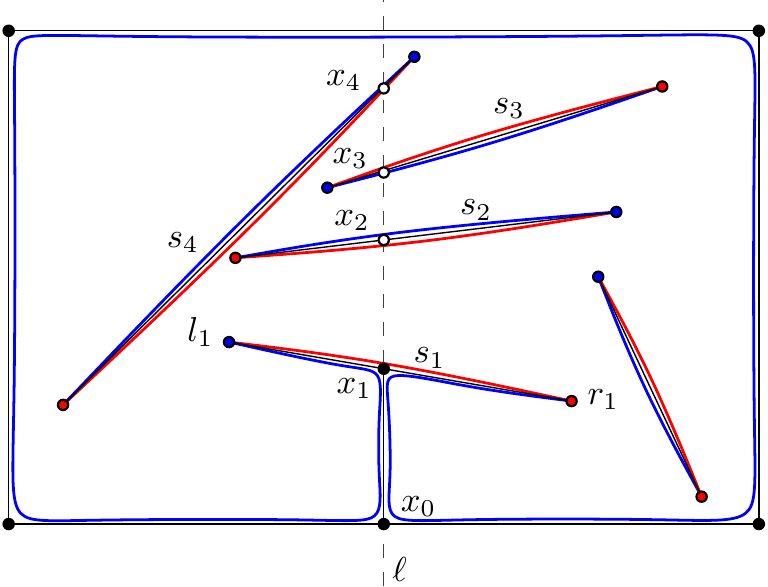}
\caption{Case where there is at least one point, lying above $s_1$, on segments $s_2,s_L$ or $s_R$ that is \cvisible with $x_1$.}
\label{fig:InputPaste}
\end{figure}

We proceed to describe the construction of the GCG that augments $M$.
Let $\mathcal R$ be a convex polygon strictly containing all segments of $M$.
Assume without loss of generality that the left endpoint of $s_1$ is blue, implying that $s_1$ is red from above and blue from below.
Let $x_0$ be the bottom intersection between $\ell$ and $\mathcal R$.
Since the bounded face of $\mathcal R$ contains no reflex vertex, we can assume that $x_0$ is blue when viewed from $x_1$. That is, $x_0$ and $x_1$ are \cvisible. 
Finally, let $X_1 = \textsc{Glue}(\mathcal R, x_1, x_0)$ be the GCG obtained by joining $x_0$ with $x_1$; see Fig.~\ref{fig:InputPaste}.

Consider the edge of $\mathcal R$ containing $x_0$ to be a segment $s_0$. 
The following invariants on the GCG $X_i$ hold initially for $i = 1$ and are maintained through a sequence of iterations.
 \begin{itemize}
\item The points $x_i$ and $x_{i+1}$ are visible while $x_i$ and $x_{i-1}$ are neighbors.
\item Besides $x_{i-1}$, vertex $x_i$ neighbors two vertices on $s_i$, one to the left and one to the right of $\ell$.
\item The endpoints of $s_i$ are reflex vertices of $X_i$.
\item The endpoints of $s_{i-1}$ are not reflex in $X_i$ and $x_{i-1}$ is an isolated vertex.
\item The color of $s_i$, when viewed from a point lying above $s_i$, is given by the color of the right endpoint of $s_i$.
\end{itemize}

Our objective is to find a point, \cvisible with $x_i$, that lies above the line extending $s_i$. If such a point exists, by gluing it to $x_i$ and then merging the remaining segments of $M$ using Lemma~\ref{lemma:Gluing the segments}, we obtain a GCG $X$ that augments $M$ with the desired properties (a full explanation is presented in Section~\ref{section:Processing after augment}, an example is depicted in Fig.~\ref{fig:InputPaste}).

As long as no such point exists, we iteratively augment $X_i$, maintaining the above properties as an invariant.
This is done with procedure $\textsc{Augment}(i)$, which takes a GCG $X_i$ and adds edges (including the edge between $x_i$ and $x_{i+1}$) to produce a new GCG $X_{i+1}$ where the above properties hold. 
After several augmentations, we will produce a GCG where the desired \cvisible point will be found.

\subsubsection*{Procedure \textsc{Augment}$(i)$}
Refer to Fig.~\ref{fig:Augment} for an illustration of this procedure.
Let $l_i$ and $r_i$ be the left and right endpoints of $s_i$, respectively.
Assume without loss of generality that $l_i$ is colored blue (and $r_i$ is red). Thus, $s_i$ is red when viewed from above.
Extend $s_i$ on both sides and let $s_L$ (\emph{resp.} $s_R$) be the first segment reached to the left (\emph{resp.} right).
This procedure is only used when the points in $s_{i+1}$, $s_L$ and $s_R$ appear blue when viewed from $x_i$. 
Otherwise, \textsc{Augment} is not required as there is a point in either $s_{i+1}$, $s_L$ or $s_R$, lying above $s_i$, that is \cvisible with $x_i$.

Notice that $s_{i+1}$, $s_L$, and $s_R$ could belong either to $M$, or to $\partial X_i$.
Let $y_L$ and $y_R$ be the points where the line extending $s_i$ intersects $s_L$ and $s_R$, respectively.
Let $X_i'$ be the PSLG obtained by adding the edges $[l_i,y_L]$ and $[r_i,y_R]$ to $X_i$ ($y_L$ and $y_R$ are added as vertices).

\begin{figure}[h!tb]
\centering
\includegraphics{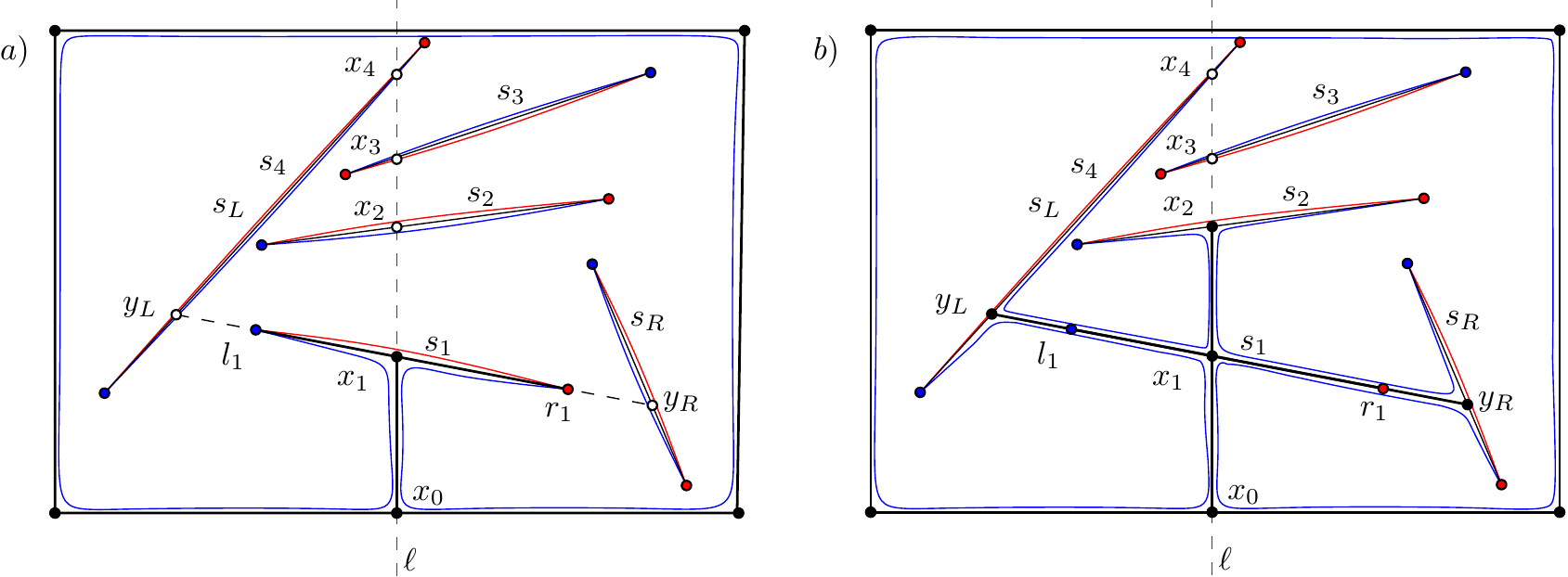}
\caption{$a)$ Example where procedure $\textsc{Augment}(1)$ is required.  Points above $s_1$ that lie on segments $s_2$, $s_L$ and $s_R$ are not \cvisible with $x_1$.
$b)$ The construction obtained by extending $s_1$, where two reflex vertices $l_1,r_1$ disappear to let $x_1$ and $x_2$ become \cvisible.}
\label{fig:Augment}
\end{figure}

This may create new faces depending on whether $s_L$ or $s_R$ belong to $M$. 
Vertices $y_L, l_i, x_i, r_i, y_R$ are collinear, meaning $l_i$ and $r_i$ are no longer reflex vertices in $X_i'$. Thus the color of $s_i$ will now be blue when viewed from above or from below.
Furthermore, if $s_L$ or $s_R$ belong to $M$, then their endpoints are now reflex vertices of $X_i'$.
One can verify that $X_i'$ is well-colored since $y_L$ and $y_R$ are both blue when viewed from $x_i$, hence $X_i'$ is a GCG. See Fig.~\ref{fig:Augment}$(b)$ for an illustration.
Notice that, when viewed from above, the color of 
$x_i$ is now blue, in contrast with the red color that $x_i$ had on $X_i$. 
Therefore, since $x_{i+1}$ is blue when viewed from below, $x_{i+1}$ and $x_i$ are now \cvisible in $X_i'$ and can be glued. 

Let $X_{i+1} = \textsc{GlueCut}(X_i', x_{i+1}, x_i)$. 
This way, the endpoints of $s_{i+1}$ become (if they were not already) reflex vertices of the GCG $X_{i+1}$ and $x_i$ becomes an isolated vertex. 
Notice that no vertex on $s_{i+1}$ neighbors a point lying above $s_{i+1}$. 
Therefore, the color of every point on segment $s_{i+1}$, when viewed from above, is given by the color of the right endpoint of $s_{i+1}$.
In fact, the invariant properties are maintained, should there be a subsequent use of $\textsc{Augment}$.

\subsection{Analysis of  AUGMENT}

\begin{observation}\label{obs:ReflexRemain}
On each iteration of $\textsc{Augment}$, all reflex vertices of $X_i$ are preserved in $X_{i+1}$, except for the two endpoints of $s_i$  that become non-reflex. 
\end{observation}

\begin{lemma}\label{lemma:AugmentFinishes}
The procedure $\textsc{Augment}$
will only be used $O(n)$ times before producing a GCG $X_j$, where there exists a point, lying above the segment $s_j$, that is \cvisible with $x_j$ (for some $1\leq j \leq k-1$).
\end{lemma}
\begin{proof}
By Lemma~\ref{lemma:ChangeInColor}, there exist segments $s_h, s_{h+1}\in S_{M,\ell}$ such that $x_h$ and $x_{h+1}$ are \cvisible before executing $\textsc{Augment}$ on $X_1$. 
We claim that $\textsc{Augment}$ can only go as far as to construct $X_h$.
If $X_h$ is not constructed, it is because a GCG $X_j$ was constructed (for some $0\leq j<h$), where there exists a point, lying above the segment $s_j$, that is \cvisible with $x_j$.
Otherwise, if $X_h$ is constructed, then, by the preserved invariants, 
the color of $x_h$, when viewed from above, remains unchanged and hence $x_h$ and $x_{h+1}$ are \cvisible.
\end{proof}

\begin{figure}[h!tb]
\centering
\includegraphics{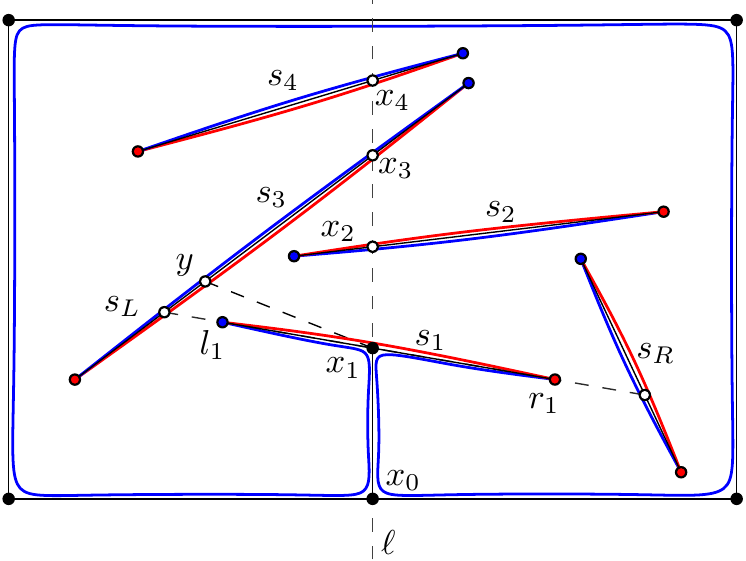}
\caption{Case where
$x_1$ and $x_2$ are not \cvisible, but a point $y$ can be found in $s_L$ so that $x_1$ and $y$ are \cvisible.}
\label{fig:AfterAugment}
\end{figure}

\subsection{Processing after  AUGMENT}\label{section:Processing after augment}

By Lemma~\ref{lemma:AugmentFinishes}, we know that after the last call to $\textsc{Augment}$ we obtain a GCG $X_j$ such that there is a point in either $s_{j+1}$, $s_L$ or $s_R$, lying above $s_j$, that is \cvisible with $x_j$. Assume without loss of generality that $x_j$ is red when viewed from above.
If $s_{j+1}$ is red when viewed from below, then $x_j$ and $x_{j+1}$ are \cvisible. 
In this case, we define $G_{M,\ell}=\textsc{GlueCut}(X_j, x_{j+1}, x_j)$.

Instead, if $x_{j+1}$ is blue when viewed from $x_j$, we follow a different approach.
Recall that the endpoints of $s_j$ are reflex vertices.
If $s_L$ is red when viewed from the left endpoint of $s_j$, choose a point $y$, slightly above $y_L$ on $s_L$, such that the whole segment $s_j$ is visible from $y$. 
Since $x_j$ is red when viewed from above, $x_j$ and $y$ are \cvisible. 
Let $G_{M,\ell}= \textsc{GlueCut}(X_j, y, x_j)$; see Fig.~\ref{fig:AfterAugment}. 
An analogous construction of $G_{M,\ell}$ follows if $s_R$ is red when viewed from the right endpoint of $s_j$. We call $G_{M,\ell}$ the \emph{extension of $X_j$}.

\begin{lemma}\label{lemma:ConstructionProperties}
If $G_{M,\ell}$ is an extension of $X_j$, then the following properties hold:

\begin{itemize}
\item The endpoints of $s_j$ are reflex vertices of $G_{M,\ell}$, but $s_j$ is not contained in $\mathcal P_{G_{M,\ell}}$.
\item The downwards ray with apex at $x_j$ does not intersect $\mathcal P_{G_{M,\ell}}$. 
\item For every $1\leq i< j$, the endpoints of $s_i$ are not reflex vertices of $G_{M,\ell}$. Moreover, $s_i$ is not contained in $\mathcal P_{G_{M,\ell}}$.
\end{itemize}
\end{lemma}
\begin{proof}
By the invariants of $\textsc{Augment}$, $x_j$ neighbors $x_{j-1}$ as well as two vertices on $s_j$, one to the left and one to the right of $\ell$.
Since $x_j$ also neighbors a vertex in $G_{M,\ell}$ lying above the segment $s_j$, $x_j$ is an isolated vertex in $G_{M,\ell}$.
Thus, by the preserved invariants and by Observation~\ref{obs:ConvexVertex},  for every $1\leq i\leq j$, $x_i$ lies outside of $\mathcal P_{G_{M,\ell}}$ and hence the segment $s_i$ is not contained in $\mathcal P_{G_{M,\ell}}$. Furthermore, the segment joining $x_i$ with $x_{i-1}$ also lies outside of $\mathcal P_{G_{M,\ell}}$ and so does the downwards ray with apex at $x_j$.
Finally, Observation~\ref{obs:ReflexRemain} tells us that, for every $1\leq i <j$, no endpoint of $s_i$ is a reflex vertex of $X_j$ (nor of $G_{M,\ell}$).
\end{proof}

We are now ready to provide the proof of Lemma~\ref{lemma:CompatibleBRMatching} which is restated for ease of readability.\\

\textsc{Lemma}~\ref{lemma:CompatibleBRMatching}.\emph{
Let $M$ be a $BR$-matching of $P$ and let $\ell$ be a chromatic cut of $M$.
There exists a $BR$-matching $M'$ of $P$, compatible with $M$, with the following properties. 
There is a segment $s$ of $M\setminus M'$ that crosses $\ell$ such that all segments of $M$ that cross $\ell$ below $s$  also belong to $M'$. Moreover, these are the only segments of $M'$ crossing $\ell$ below $s$.
}

\begin{proof}
Let $G_{M,\ell}$ be the GCG obtained using the construction presented in this section on $M$ and $\ell$.
Recall that $S_{M,\ell}$ is the set of segments of $M$ that cross $\ell$.
Lemma~\ref{lemma:ConstructionProperties} states that there is a segment $s_j\in S_{M,\ell}$, such that its endpoints are reflex vertices of $G_{M,\ell}$ but $s_j$ is not contained in $\mathcal P_{G_{M,\ell}}$.
Let $W$ be the set of segments in $M$ that are contained in the interior of $G_{M,\ell}$ and let $Z_\ell = \{s_1, \ldots, s_{j-1}\}$ be the set of segments of $S_{M,\ell}$ that cross $\ell$ below $x_j$. 
By Lemma~\ref{lemma:ConstructionProperties} we know that $Z_\ell\cap W = \emptyset$. 

By Lemma~\ref{lemma:Gluing the segments}, since $W$ is contained in the interior of $G_{M,\ell}$, we can augment $G_{M,\ell}$ by gluing the segments of $W$ to its boundary such that the endpoints of every segment in $W$ become reflex vertices in $G_{M,\ell}$. Moreover, the reflex vertices of $G_{M,\ell}$ are preserved. 

By Lemma~\ref{lemma:ReflexMatching}, there exists a $BR$-matching $W'$ of the reflex vertices of $G_{M,\ell}$ such that each segment in $W'$ is contained in $\mathcal P_{G_{M,\ell}}$. 
Notice that the endpoints of $s_j$ are re-matched in $W'$. However, since $s_j$ is not contained in $\mathcal P_{G_{M,\ell}}$, $s_j$ does not belong to $W'$. 
Moreover, Lemma~\ref{lemma:ConstructionProperties} implies that the ray, shooting downwards from $x_j$, lies outside $\mathcal P_{G_{M,\ell}}$. 
Thus, no segment in $W'$ crosses $\ell$ below $x_j$.

Let $M' = W'\cup Z_\ell$ be a set of bichromatic segments. 
Every point in $P$ is matched in $M'$ since every point in $P$ is either a reflex vertex of $G_{M,\ell}$, or an endpoint of a segment in $Z_\ell$.
Lemma~\ref{lemma:ConstructionProperties} implies that the endpoints of the segments in $Z_\ell$ are not reflex vertices in $G_{M,\ell}$. Therefore, $M'$ is a $BR$-matching of $P$. 
Since $W$ and $W'$ are compatible, $M$ and $M'$ are compatible $BR$-matchings.
\end{proof}

\section{Remarks}

Although the techniques developed in this paper appear tailored for this specific problem, they have a more general underlying scope.
At a deeper level, our tools generate a \emph{balanced} convex partition of the plane. Roughly speaking, in Lemma~\ref{lemma:AbellanasAlgorithm} a simple polygon is partitioned into a set of convex polygons obtained by shooting rays from the reflex vertices towards the interior of the polygon until hitting its boundary. Once the polygon is partitioned, a matching can be found on each convex piece.
By using Lemma~\ref{lemma:AbellanasAlgorithm} in the bichromatic setting, we generate a convex partition of the GCGs, where each convex face is in charge of matching a balanced number of red and blue points. Convex partitions, usually constructed by extending segments until they reach another segment or a previously extended section, have been extensively used to solve several augmentation and reconfiguration problems~\cite{Aichholzer2009617, CompatibleMatchingsForPSLG, Ishaque2011, Hurtado200814, Hoffmann201035}. Therefore, the techniques provided in this paper are of independent interest. In conjunction with Lemma~\ref{lemma:ReflexMatching}, operators like \textsc{Glue} and \textsc{Cut} can be used to find special convex partitions that provide new ways to construct compatible PSLGs.\\

\textbf{Acknowledgements. }
We thank the authors of~\cite{CompatibleMatchingsForPSLG} for highlighting
the open problem from their paper to the audience of the EuroGIGA meeting that took place after EuroCG 2012.
We would also like to thank Tillmann Miltzow for useful comments.

This research was done during a visit of Diane L. Souvaine to ULB supported by FNRS.

{\small
\bibliographystyle{abbrv}
\bibliography{DichromaticMatchings}
}

\end{document}